\pdfoutput=1
\documentclass[a4paper,10pt]{llncs}
\usepackage{latexsym,epsfig,amsmath,amssymb}
\usepackage{graphics,subfigure,paralist}
\usepackage{graphicx}
\usepackage{fleqn,latexsym}
\usepackage[small]{caption}
\usepackage[]{inputenc}
\usepackage[T1]{fontenc}
\usepackage{xspace,hyperref}
\usepackage{coqdoc}
\hypersetup{citecolor=black,linkcolor=black}
\usepackage{times}

\newlength{\hsbw}
\newcommand\HOLSpacing{13pt}

\def\R{{\ooalign{\hfil\raise.2ex\hbox{\tiny R}\hfil\small\crcr\mathhexbox20D}}}

\DeclareSymbolFont{AMSb}{U}{msb}{m}{n}
\DeclareMathSymbol{\N}{\mathbin}{AMSb}{"4E}
\DeclareMathSymbol{\B}{\mathbin}{AMSb}{"4E}

\renewcommand{\implies}{\supset}

\def\b2STE{\mathit{bool2ste}}

\renewcommand{\Cup}{\bowtie}

\def\imp{\supset}

\def\Val{\mathit{Eval}}
\def\AND{And}

\newcommand{\cunsat}{\textrm{SHRUTI}\xspace}
\newenvironment{itemizeC}{\begin{compactenum}}{\end{compactenum}}

\begin{document}

\title{Industrial-Strength Formally Certified\\ SAT Solving}

\author{Ashish Darbari\inst{1} and Bernd Fischer\inst{1} and Joao Marques-Silva\inst{2}}

\institute{{School of Electronics and Computer Science\\
    University of Southampton, Southampton, SO17 1BJ, England}
  \and
  {School of Computer Science and Informatics\\
    University College Dublin, Belfield, Dublin 4, Ireland}}


\maketitle

\begin{abstract}
Boolean Satisfiability (SAT) solvers are now routinely used in the
verification of large industrial problems. However, their application
in safety-critical domains such as the railways, avionics, and
automotive industries requires some form of assurance for the results,
as the solvers can (and sometimes do) have bugs. Unfortunately, the
complexity of modern, highly optimized SAT solvers renders impractical
the development of direct formal proofs of their correctness. This
paper presents an alternative approach where an untrusted,
industrial-strength, SAT solver is plugged into a trusted, formally
certified, SAT proof checker to provide industrial-strength certified
SAT solving. The key novelties and characteristics of our approach are
(i) that the checker is automatically extracted from the formal
development, (ii), that the combined system can be used as a
standalone executable program independent of any supporting theorem
prover, and (iii) that the checker certifies any SAT solver respecting
the agreed format for satisfiability and unsatisfiability claims.  The
core of the system is a certified checker for unsatisfiability claims
that is formally designed and verified in Coq.  We present its formal
design and outline the correctness proofs. The actual standalone
checker is automatically extracted from the the Coq development.
An evaluation of the certified checker on a representative set of
industrial benchmarks from the SAT Race Competition shows that, albeit
it is slower than uncertified SAT checkers, it is significantly faster
than certified checkers implemented on top of an interactive theorem
prover.
\end{abstract}

\section{Introduction}
Advances in Boolean satisfiability SAT technology have made it possible
for SAT solvers to be routinely used in the verification of large
industrial problems, including safety-critical domains that require a
high degree of assurance such as the railways, avionics, and
automotive industries~\cite{Hammarberg:2005p212,Penicka:2007p6353}.
However, the use of SAT solvers in such domains requires some form of
assurance for the results. This assurance can be provided in two
different ways.



First, the solver can be proven correct once and for all. However,
this approach had limited success.
For example,
Lescuyer et al.~\cite{Lescuyer:2009p1671} 
formally designed and verified
a SAT solver using the Coq proof-assistant \cite{CoqBook}, but without
any of the techniques and optimizations used in modern solvers. Reasoning about these
optimizations makes the formal correctness proofs exceedingly hard. This was
shown by the work of Mari{\'c}~\cite{Maric:2009p6673}, who verified the
algorithm 
used in the ARGO-SAT solver but restricted the verification to the pseudo-code level, and in particular, 
did not verify the actual solver itself. 
In addition, the formal verification has to be repeated for every new SAT
solver (or even a new version of a solver), or else the user is locked into
using the specific verified solver.


Alternatively, a \emph{proof checker} can be used to validate each
individual outcome of the solver independently; this requires
the solver to produce a \emph{proof trace} that is viewed as a
certificate justifying  the outcome of the solver.
This approach 
was used to design several checkers such as tts, 
Booleforce, PicoSAT and zChaff~\cite{satcomp}.
However, these checkers are typically implemented by the developers
of the SAT solvers whose output they are meant to check, which can lead
to bugs being masked, and none of them was formally designed or
verified, which means that they provide only limited assurance.




The problems of both approaches can be circumvented if the \emph{checker} rather than
the solver is proven correct, once and for all. This is
substantially simpler than proving the solver correct, because the
checker is comparatively small and straightforward. It does not lead to a
system lock-in, because the checker can work for all solvers that can
produce proof traces (certificates) in the agreed format.  
This approach was followed by Weber and Amjad \cite{Weber:2009p867} 
in their formal development
of a proof checker for zChaff and Minisat proof traces. Their core idea 
is to replay the derivation encoded in the proof trace
\emph{inside} LCF-style interactive theorem provers such as HOL~4, Isabelle, and HOL Light. Since the
design and implementation of these provers is based on a small trusted
kernel of inference rules, assurance is very high. However this comes
at the cost of usability: their checker can run \emph{only inside} the
supporting prover, and not as a standalone tool. Moreover, performance
bottlenecks become prominent when the size of the problems increases.
\begin{figure}[t]
\begin{center}
  \includegraphics[scale=1.2]{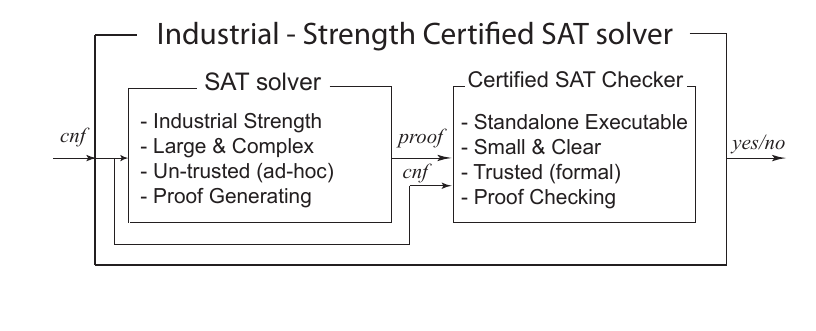}
\end{center}\vspace*{-10mm}
\caption{High-Level View of an Industrial-Strength Formally Certified SAT Solver.}\label{cunsat}
\end{figure}

Here, we follow the same general idea of a formally certified proof
checker, but depart considerably from Weber and Amjad
in how we design and implement our solution.  
We describe an approach where one can plug an untrusted,
industrial-strength SAT solver into a formally certified SAT proof
checker to provide an industrial-strength certified SAT solver.
We designed, formalized and verified the SAT proof checker for both
satisfiable and unsatisfiable problems. In this paper, 
we focus on the more interesting aspect of checking
unsatisfiable claims; satisfiable certificates are significantly
easier to formally certify.
Our certified checker \cunsat  is formally designed and verified using
the higher-order logic  based  proof assistant Coq~\cite{CoqBook}, but
we never use Coq as a checker; instead we {\it automatically} extract an OCaml program
from the formal development that is compiled to a standalone
executable that is used independently of Coq. In this regard, our
approach prevents the user to be locked-in to a specific proof
assistant, something that was not possible with Weber and Amjad's
approach.
 A high-level
architectural view of our approach is shown in Figure~\ref{cunsat}.
Since it combines certification and ease-of-use, it enables the use of
certified checkers as regular components in a SAT-based
verification work flow.





\section{Propositional Satisfiability}

\subsection{Satisfiability Solving}
Given a propositional formula, the
goal of satisfiability solving is to determine whether there is 
an assignment of the Boolean truth values (i.e., True, False) to the
variables in the formula such that the formula evaluates to true. If such an
assignment exists, the given formula is said to be \emph{satisfiable} or SAT,
otherwise the formula is said to be \emph{unsatisfiable} or
UNSAT. Many problems of practical interest in system verification
involved proving unsatisfiability, one concrete example being bounded
model checking~\cite{biere-bmc03}.


For efficiency purposes, SAT solvers represent the propositional formulas in
conjunctive normal form (CNF), where the entire formula is a conjunction of
\emph{clauses}. Each clause itself denotes a disjunction of \emph{literals},
which are simply (Boolean) variables or negated variables.
An efficient CNF representation uses non-zero integers to represent literals.
A positive literal is represented by a
positive integer, whilst a negated one is denoted by a negative integer. 
As an example, the (unsatisfiable) formula 
\[(a \wedge b) \vee (\neg a \wedge b) \vee (a \wedge \neg b) \vee (\neg a \wedge \neg b)\]
over two propositional variables $a$ and $b$
can thus be represented as

\smallskip
\begin{tabular}{rrr}  
  1\; & 2\; & 0\\
  -1\; & 2\; & 0\\
  1\; & -2\; & 0\\
  -1\; & -2\; & 0
\end{tabular}
\smallskip
\coqdoceol

The zeroes are delimiters that separate the clauses from each other.



SAT solvers take a Boolean formula, for example represented in the 
DIMACS notation used here,
and produce a SAT/UNSAT claim. 
A \emph{proof-generating} SAT solver produces additional evidence (or
\emph{certificates}) to support its claims. For a SAT claim, the certificate
simply consists of an assignment.
It is
usually trivial to check whether that assignment---and thus
the original SAT claim---is correct. One simply
substitutes the Boolean values given by the assignment in the formula
and then evaluates the overall formula, checking that it indeed is true.
For UNSAT claims, the evidence is more complicated, and the solvers
need to return a resolution \emph{proof trace} as certificate.
Unsurprisingly, checking these UNSAT certificates is more complicated as well.

\subsection{Proof Checking}
When the solver claims a given problem is UNSAT, we can independently
re-play the proofs produced by the solver, to check that the solver's
output is correct. For a given problem, if we can follow the
resolution inferences given in the proof trace to derive an empty clause, then
we know that the proof trace correctly denotes an UNSAT instance for
the problem, and we can conclude that the given problem is indeed
UNSAT.

A proof trace consists of the original clauses used during resolution and
the intermediate resolvents obtained by resolving the original input
clauses. 
The part of the proof trace that specifies how the input clauses have been resolved in sequence to derive a conflict is organized as {\it chains}. These
chains are often referred to as the regular input resolution proofs,
 or the trivial proofs~\cite{Beame:2004p1080,Biere:2008p301}. We call
 the input clauses in a chain its \emph{antecedents} and its final
 resolvent simply its \emph{resolvent}. Designing an efficient checking
methodology relies to some extent on the representation of the proof
trace produced by the solvers. Representing proof chains as a trivial
resolution proof is a key constraint~\cite{Biere:2008p301}.  The trivial resolution
proof requires that the clauses generated to form the proof should be
aligned in such a manner in the trace so that whenever a pair of
clauses is used for resolution, {\it at most} one complementary pair of
literals is deleted. Another important constraint that is needed for efficiency
reasons is that {\it at least} one pair of complementary literals gets
deleted whenever any two clauses are used in the trace for
resolution. This is needed because we want to avoid search during
checking, and if the proof trace respects these two criteria we can
design a checking algorithm that can validate the proofs in a single
pass using an algorithm which is linear in the size of the input
clauses. 

\subsection{PicoSAT Proof Representation}


Most proof-generating SAT
solvers~\cite{Biere:2008p301,minisat,zhang-date03} preserve these two
criterions. We carried out our first set of experiments with
PicoSAT~\cite{Biere:2008p301}. PicoSAT ranked as one of the best
solvers in the industrial category of the SAT Competitions 2007 and
2009, and in the SAT Race 2008. Moreover, PicoSAT's proof representation
is in {\it ASCII} format. It was reported by Weber and
Amjad~\cite{Weber:2009p1923} that some other solvers such as Minisat
generate proofs in a more compact, but binary format.

Another advantage of PicoSAT's proof trace format besides using an
ASCII format is its simplicity. It does not record the information
about pivot literals as some of the other solvers such as
zChaff. It is however straightforward to develop translators
for formats used by other SAT solvers~\cite{vangelder-sat07}.
As a proof of concept, we developed a translator from
zChaff's proof format to PicoSAT's proof format, which can be used for
validating unsatisfiability proofs obtained with zChaff.
%

A PicoSAT proof trace consists of
 rows representing the input clauses, followed by rows encoding the proof
chains. Each row representing a chain consists of an asterisk (*) as place-holder 
for the chain's resolvent,%
\footnote{This is generated by PicoSAT; there is another option of
  generating proof traces from PicoSAT where instead of the asterisk
  the actual resolvents are generated delimited by a single zero from the
  rest of the chain.
} 
followed by the identifiers of the clauses involved in the chain. Each chain
row thus contains at least two clause identifiers, and denotes an
application of one or more of the resolution inference rule, describing a
trivial resolution derivation. Each row also starts with a non-zero positive 
integer denoting the identifier for that row's (input or resolvent) clause. 
In an actual trace there are additional zeroes as delimiters at the
end of each row, but we remove these before we start proof
checking. For the UNSAT formula shown in the previous section, the
corresponding proof trace generated from PicoSAT looks as follows:

\smallskip
\begin{tabular}{rrrrr}
  1\; & 1\; & 2\; &  \\
  2\; & -1\; & 2\; &  \\
  3\; & 1\; & -2\; &  \\
  4\; & -1\; & -2\; &  \\
  5\; & *\; & 3\; & 1\; \\
  6\; & *\; & 4\; & 2\; & 5 \\
\end{tabular}
\smallskip

The first four rows denote the input clauses from the original problem (see
above) that are used in the resolution, with their identifiers referring to the
original clause numbering, whereas rows 5 and 6 represent the proof chains.  In
row 5, the clauses with identifiers 3 and 1 are resolved using a single
resolution rule, whilst in row 6 first the original clauses with identifier 4
and 2 are resolved and then the resulting clause is resolved against the clause
denoted by identifier 5 (i.e., the resolvent from the previous chain), in total
using two resolution steps. 

PicoSAT by default creates a compacted form of proof traces, where 
the antecedents for the derived clauses are
not ordered properly within the chain. This means that there are
instances in the chain where we resolve a pair of adjacent clauses and
no literal is deleted. This violates one of the constraints we
explained above, thus we cannot deduce an existence of an empty clause
for this trace unless we re-order the antecedents in the chain.

PicoSAT comes with an uncertified proof checker called Tracecheck that
can not only check the outcome of PicoSAT but also corrects the
mis-alignment of traces. The outcome of the alignment process is an
extended proof trace and these then become the input to the certified
checker that we design.

\section{The SHRUTI Certified Proof Checker}
Our approach to efficient formally certified SAT solving relies on the
use of a \emph{certified} checker program that has been designed and
implemented formally, but can be used in practice, independently of any formal
development environment. Rather than verifying a separately developed
checker we follow a correct-by-construction approach in which we
formally design and mechanically verify a checker using a 
proof assistant to achieve the highest level of confidence,
and then use program extraction to obtain a standalone executable.

Our checker \cunsat 
takes an input CNF file
which contains the original problem description and a proof
trace file and checks whether the two together denote an (i) UNSAT
instance, and (ii) that they are ``legitimate''. Our focus here is on
(i) where we check that each step of the trace is correctly applying the resolution
inference rule. However, in order to gain full assurance about the UNSAT
claim, we also need to check that all input clauses used in the resolution
in the trace are taken from the original problem, i.e., that the proof 
trace is legitimate. \cunsat provides this as an option, but
as far as we are aware most checkers do not do this check.


Our formal development follows the LCF style~\cite{lcf},
and in particular, only uses definitional extensions, i.e., new
theorems can only be derived by applying previously derived inference rules. 
Axiomatic extensions though possible are prohibited, since one can
assume the existence of a theorem without a proof. Thus, we never use
it in our own work. 
We use the the Coq proof assistant~\cite{CoqBook} as a development tool.
Coq is based on the Calculus of Inductive Constructions
and encapsulates the concepts of typed higher-order logic. 
It uses the notion of proofs as types, and 
allows constructive proofs and use of dependent types. 
It has been successfully used in the design and
implementation of large scale certification of software such as in the
CompCert~\cite{Leroy:2008p6143} project.

For our development of \cunsat, we first formalize in Coq the definitions of
resolution and its auxiliary functions and then prove inside Coq that these 
definitions are correct, i.e., satisfy the properties expected of resolution. 
Once the Coq formalization is complete, OCaml code is extracted from it 
through the extraction API included in the Coq proof assistant. 
The extracted OCaml code expects
its input in data structures such as tables and lists. These data structures
are built by some glue code that also handles file I/O and pre-processes the
read proof traces (e.g., removes the zeroes used as separators for the 
clauses). The glue code wraps around the extracted checker and
the result is then compiled to a native machine code executable that
can be run independently of the proof-assistant Coq.

\subsection{Formalization in Coq}
In this section we present the formalization of \cunsat. Its core logic
is formalized as a shallow embedding in Coq. In a shallow embedding we
identify the object data types (types used for \cunsat)
 with the types of the meta-language, which in our case happens to be the
 Coq datatypes. 


Since most checkers read the clause
representation on integers (e.g., using the DIMACS notation) we incorporate
integers as first class elements in our formalization, so we do not
have to map literals to Booleans. Thus, inside Coq, we
denote {\it literals} by integers, and {\it clauses} by lists of
integers. Antecedents (denoting the input clauses) in a proof chain are
represented by integers and a proof chain itself by a list of integers. 
A resolution proof is represented internally in our
implementation by a table consisting of (\emph{key},\emph{binding})
pairs.  The key for this table is the identifier obtained from the
proof chains read from the input proof trace file. The binding of this
table is the actual resolvent obtained by resolving the clauses
specified (in the proof trace). When the input proof trace is read,
the identifier corresponding to the first row of the proof chain
becomes the starting point for resolution checking. Once the resolvent
is calculated for this the process is repeated for all the remaining
rows of the proof chain, until we reach the end of the trace
input. If the identifier for the last row of the proof chain
denotes an empty resolvent, we conclude that the given problem and its
trace represents an UNSAT instance.


We use the usual notation for quantifiers ($\forall$, $\exists$) and
logical connectives ($\wedge, \vee, \neg$) but
distinguish implication over propositions ($\imp$) and
over types ($\rightarrow$) for presentation clarity, though inside
Coq they are exactly the same. The notation $\Rightarrow$ is used
during pattern matching (using \coqdockw{match-with}) as in other functional languages. For type
annotation we use $\mathrm{:}$, and for the cons operation on lists 
we use \coqdocvar{::}. The empty list is denoted by \coqdocvar{nil}. The
set of integers is denoted by \coqdocvar{Z}, the type of polymorphic
list by  \coqdocvar{list} and the list of integers by
\coqdocvar{list\, Z}. List containment is represented by $\in$ and its
negation by $\notin$. The function $\mathit{abs}$ computes the
absolute value of an integer.  We use the keyword
\coqdockw{Definition} to present our function
definitions. It is akin to defining functions in Coq. Main data structures that we used in the Coq formalization are lists, 
and finite maps (hashtables with integer keys and polymorphic bindings).

We define our resolution function ($\bowtie$) with the help of two auxiliary
functions $\mathit{union}$ and $\mathit{auxunion}$. 
Both functions compute the union of two clauses represented as integer lists,
but differ in their behavior when they encounter complementary literals:
whereas $\mathit{union}$ \emph{removes} both literals and then calls 
$\mathit{auxunion}$ to process the remainder of the lists;
$\mathit{auxunion}$ \emph{copies} both
the literals into the output and thus produces a tautological clause.
Ideally, if the SAT solver is sound and the proof trace reflects the sound
outcome, then for any pair of clauses that are resolved, there will be
only one pair of complementary literals and we do not need two
functions. However in reality, a solver or its proof trace can have
bugs and it can create instances of clauses in the trace with multiple
complementary pair of literals.  Hence, we employ the two auxiliary functions
to ensure that the resolution function deals with this in a sound way.

We will later explain in more detail the functionality of the auxiliary
functions but 
both functions 
expect the input clauses to respect three well-formedness
criteria:  there
should be no duplicates in the clauses ({$\mathit{NoDup}$}\normalsize); there
should be no complementary pair of literals {\it within} any clause
({$\mathit{NoCompPair}$), and the clauses should be sorted by absolute value
($\mathit{sorted}$).
The predicate {$\mathit{Wf}$} \normalsize encapsulates these properties.
\coqdocemptyline
\hrule
\coqdocemptyline
 {$\mathsf{Definition}\;\; \mathit{Wf}\; c=\; 
  \mathit{NoCompPair}\, c\; \wedge\; \mathit{NoDup}\,c\; \wedge\;
  \mathit{sorted}\, c$}
\coqdocemptyline
\hrule
\coqdocemptyline
The assumptions that there are no duplicates and no complementary pair
of literal within a clause are essentially the constraints imposed on
input clauses when the resolution function is applied in
practice. Sorting is enforced by us to keep the complexity of our
algorithm linear. 


The union function takes a pair of sorted (by absolute value) lists of
integers, and an accumulator list, and computes the
resolvent by doing a pointwise comparison on input
literals. 
\coqdocemptyline
\hrule
\coqdocemptyline
\coqdocnoindent
\coqdockw{Definition} \coqdocvar{union} (\coqdocvar{c_1}
\coqdocvar{c_2} : \coqdocvar{list\, Z})(\coqdocvar{acc} :
\coqdocvar{list\, Z}) = \coqdoceol
\coqdocnoindent
\coqdockw{match} \coqdocvar{c_1,c_2} \coqdockw{with}\coqdoceol
\coqdocindent{0.25em}
\ensuremath{|} \coqdocvar{nil,c_2}   \ensuremath{\Rightarrow} \coqdocvar{app} (\coqdocvar{rev} \coqdocvar{acc})  \coqdocvar{c_2}\coqdoceol
\coqdocindent{0.25em}
\ensuremath{|} \coqdocvar{c_1,nil}   \ensuremath{\Rightarrow} \coqdocvar{app} (\coqdocvar{rev} \coqdocvar{acc})  \coqdocvar{c_1}\coqdoceol
\coqdocindent{0.25em}
\ensuremath{|} \coqdocvar{x::xs}, \coqdocvar{y::ys}
\ensuremath{\Rightarrow}
\coqdockw{if} (\coqdocvar{x} + \coqdocvar{y} = {0}) \coqdockw{then} 
\coqdocvar{auxunion} \coqdocvar{xs} \coqdocvar{ys} \coqdocvar{acc}
\coqdoceol
\coqdocindent{8.3em}
\coqdockw{else} \coqdockw{if}  (\coqdocvar{abs} \coqdocvar{x} $<$
\coqdocvar{abs} \coqdocvar{y}) \coqdockw{then}
\coqdocvar{union} \coqdocvar{xs} (\coqdocvar{y::ys})(\coqdocvar{x::acc}) 
\coqdoceol
\coqdocindent{8.3em}
\coqdockw{else} 
\coqdockw{if}  (\coqdocvar{abs} \coqdocvar{y} $<$ \coqdocvar{abs}
\coqdocvar{x}) \coqdockw{then} 
\coqdocvar{union} (\coqdocvar{x::xs}) \coqdocvar{ys} (\coqdocvar{y::acc}) 
\coqdoceol
\coqdocindent{8.3em}
\coqdockw{else}
\coqdocvar{union} \coqdocvar{xs} \coqdocvar{ys} (\coqdocvar{x::acc})\coqdoceol
\coqdocemptyline
\hrule
\coqdoceol


We already pointed out above that this function and the auxiliary
union function - $\mathit{auxunion}$ (shown below) that it employs are
different in behaviour if the literals being compared are
complementary. 
However, when the literals are non-complementary, if they are equal, only one
copy is put in the resolvent whilst when they are unequal both are
kept in the resolvent.  Once a single run of any
of the clauses is finished, the other clause is merged with 
the accumulator. Actual sorting  in our case is done by simply reversing the
 accumulator (since all elements are in descending order).

\coqdocemptyline
\hrule
\coqdocemptyline
\coqdocnoindent
\coqdockw{Definition} \coqdocvar{auxunion} (\coqdocvar{c_1}
\coqdocvar{c_2} : \coqdocvar{list\, Z})(\coqdocvar{acc} :
\coqdocvar{list\, Z}) = \coqdoceol
\coqdocnoindent
\coqdockw{match} \coqdocvar{c_1,c_2} \coqdockw{with}\coqdoceol
\coqdocindent{0.25em}
\ensuremath{|} \coqdocvar{nil,c_2}   \ensuremath{\Rightarrow} \coqdocvar{app} (\coqdocvar{rev} \coqdocvar{acc})  \coqdocvar{c_2}\coqdoceol
\coqdocindent{0.25em}
\ensuremath{|} \coqdocvar{c_1,nil}   \ensuremath{\Rightarrow} \coqdocvar{app} (\coqdocvar{rev} \coqdocvar{acc})  \coqdocvar{c_1}\coqdoceol
\coqdocindent{0.25em}
\ensuremath{|} \coqdocvar{x::xs}, \coqdocvar{y::ys}
\ensuremath{\Rightarrow}
\coqdockw{if}  (\coqdocvar{abs} \coqdocvar{x} $<$ \coqdocvar{abs} \coqdocvar{y}) \coqdockw{then}
\coqdocvar{auxunion} \coqdocvar{xs} (\coqdocvar{y::ys})
(\coqdocvar{x::acc}) 
\coqdoceol
\coqdocindent{8.3em}
\coqdockw{else} \coqdockw{if}  (\coqdocvar{abs} \coqdocvar{y} $<$
\coqdocvar{abs} \coqdocvar{x}) \coqdockw{then} 
\coqdocvar{auxunion} (\coqdocvar{x::xs}) \coqdocvar{ys} (\coqdocvar{y::acc})\coqdoceol\coqdocindent{8.3em}
\coqdockw{else} \coqdockw{if} \coqdocvar{x}=\coqdocvar{y} \coqdockw{then}
\coqdocvar{auxunion} \coqdocvar{xs} \coqdocvar{ys} (\coqdocvar{x::acc})
\coqdoceol
\coqdocindent{8.3em}
\coqdockw{else}
\coqdocvar{auxunion} \coqdocvar{xs} \coqdocvar{ys} (\coqdocvar{x::y::acc})
\coqdocemptyline
\hrule
\coqdoceol

We can now show the actual resolution function denoted by
$\Cup$ below. It makes use of the $\mathit{union}$ function.
\coqdocemptyline
\hrule
\coqdocemptyline
\coqdocnoindent
\coqdockw{Definition} \coqdocvar{c_1} $\Cup$ \coqdocvar{c_2} = (\coqdocvar{union} \coqdocvar{c_1} \coqdocvar{c_2} \coqdocvar{nil})\coqdoceol
\coqdocemptyline
\hrule
\coqdocemptyline

Given a problem representation in CNF form and a proof trace that respects the
well-formedness criterion and is a trivial resolution proof for the given
problem, \cunsat will deduce the empty clause and thus validate the solver's
UNSAT claim.  Conversely, whenever \cunsat validates a claim, the problem is
indeed UNSAT---\cunsat will never deduce an empty clause for a SAT instance and
will thus never give a false positive. 

If \cunsat cannot deduce the empty clause, it invalidates the claim.  This
situation can be caused by three different reasons. First, counter to the claim
the problem is SAT. Second, the problem may be UNSAT but the 
resolution proof may not represent this because it may have bugs. Third,
the traces either do not respect the well-formedness criteria, or do not
represent a trivial resolution proof.  In this case both the problem and the
proof may represent an UNSAT instance but our checker cannot validate it as
such.  In this respect our checker is incomplete.


\subsection{Soundness of the resolution function}
Here we formalize the soundness criteria for our
checker and present the soundness theorem stating that the definition
of our resolution function is sound.
We need to prove that the resolvent of a given pair of clauses is
logically entailed by the two clauses. Thus at a high-level we need to prove that:

\[\forall c_1\, c_2\, c_3\cdot\, (c_3 = c_1 \bowtie c_2) \implies\;
\{c_1,c_2\} \models\ c_3\]

\noindent
where $\models$ denotes the logical entailment. 

%
We can use the deduction theorem 
\[\forall a\, b\, c\cdot\, \{a,b\} \models c\; \equiv\; (a \wedge b \implies c) \]
to re-state what we
intuitively would like to prove:

\[\forall c_1\, c_2\cdot\,(c_1 \wedge c_2) \implies\ (c_1 \bowtie c_2)\]

\noindent
However instead of proving the above theorem directly we prove its contraposition:

\[\forall c_1 c_2\cdot\, \neg(c_1 \bowtie c_2)\; \implies\; \neg(c_1 \wedge c_2)\]

In our formalization a clause is denoted by a list of non-zero
integers, and a conjunction of clauses is denoted by a list of
clauses. We now present the definition of the logical
disjunction and conjunction functions that operate on the integer list
representation of clauses. We do this with the help of an
interpretation function that maps an integer to a Boolean value. 

The function $\Val$ shown below maps a list of
integers to a Boolean value by using the logical disjunction $\vee$
and an interpretation $I$ of the type $Z\rightarrow Bool$.
\coqdocemptyline
\hrule
\coqdocemptyline
\ensuremath{
\coqdockw{Definition}\; \Val\; nil\; I\; =\; \mathit{False} \\
\hspace*{20mm}\Val\;(x::xs)\; I\; =\; \Val\; x\; I\; \vee\; (\Val\; xs\; I)
}
\coqdocemptyline
\hrule
\coqdocemptyline
We now define what it means to perform a conjunction over a list of
 clauses. The function $\AND$ shown below takes a list of clauses and
 an interpretation $I$ (with type $Z \rightarrow Bool$) and returns a
Boolean which denotes the conjunction of all the clauses in the
list.
\coqdocemptyline
\hrule
\coqdocemptyline
\ensuremath{
\coqdockw{Definition}\; \AND\; nil\; I\; =\;\mathit{True}\\
\hspace*{20mm} \AND\; (x::xs)\; I\; =\; (\Val\, x\; I)\; \wedge\; (\AND\; xs\; I)
}
\coqdocemptyline
\hrule
\coqdocemptyline
The interpretations that we are interested in are the
logical interpretations which means that if we apply an interpretation I
on a negative integer the value returned is the logical negation of
the value returned when the same interpretation is applied on the
positive integer.
\coqdocemptyline
\hrule
\coqdocemptyline
$\coqdockw{Definition}\;\; Logical\; I = \forall
  (x:Z)\cdot\, I (-x) = \neg (I\; x)$
\coqdocemptyline
\hrule
\coqdocemptyline
Thus we can now state the precise statement of the soundness theorem
that we proved for our checker as:

\begin{theorem}{Soundness theorem}\\
\hspace*{5mm}$\forall c_1 c_2\cdot\, \forall\, I\cdot\,
Logical\; I\;\implies\;\neg(\Val\;(c_1 \bowtie c_2)\; I)\;
\implies\; \neg(\AND\; [c_1,c_2]\;I)$
\end{theorem}
\begin{proof}
The proof begins by structural induction on $c_1$ and $c_2$. The first
three sub-goals are easily proven by term rewriting and simplification
by unfolding the definitions of $\bowtie$, $\Val$ and $\AND$. The last
sub-goal is proven by doing a case split on if-then-else and then
using a combination of induction hypothesis and generating conflict
among some of the assumptions. A detailed transcription of the Coq proof is
available from http://sites.google.com/site/certifiedsat/.
\end{proof}



\subsection{Correctness of Implementation}
Further to ensure that the formalization of our checker is correct we 
check that the union function is implemented correctly. We check that
it preserves the following properties of the trivial resolution
function. These properties are:
\begin{itemizeC}
\item A pair of complementary literals is deleted in the resolvent
obtained from resolving a given pair of clauses (Theorem~\ref{compl}).

\item All non-complementary pair of literals that are unequal 
are retained in the resolvent (Theorem~\ref{nocompl}). 

\item For a given pair of clauses, if there are no duplicate literals
  within each clause, then for a literal that exists in both the
  clauses of the pair, only one copy of the literal is retained in the
  resolvent (Theorem~\ref{equal}).

\end{itemizeC}
We have proven these properties in Coq. The actual proof consists of
proving several small and big lemmas - in total about 4000 lines of
proof script in Coq (see the proofs online).

The general strategy  is to use structural
induction on clauses \coqdocvar{c_1} and \coqdocvar{c_2}. For each theorem, 
this results
in four main goals, three of which are proven by contradiction since
for all elements $\ell$,  $\ell \notin nil$. For the remaining goal 
a case-split is done on if-then-else, thereby producing sub-goals,
some of whom are proven from induction hypotheses, and some from conflicting
assumptions arising from the case-split. 

\begin{theorem}{A pair of complementary literals is deleted.}\label{compl}
\coqdoceol
\coqdocnoindent
\ensuremath{
\hspace*{5mm}\forall \coqdocvar{c_1}\ \coqdocvar{c_2}\cdot\ 
\mathit{Wf}\ c_1\ \imp\ \mathit{Wf}\ c_2\ \imp \mathit{UniqueCompPair}\; c_1\,
c_2\ \imp \\
\hspace*{20mm}    \forall {\ell}_1\ \ell_{2}\cdot\
      (\ell_1 \in \coqdocvar{c_1})\ \imp\ (\ell_2 \in \coqdocvar{c_2})\ 
      \imp\ (\ell_1+\ell_2=0) \imp \\ 
\hspace*{3.5cm}(\ell_1 \notin (\coqdocvar{c_1} \Cup \coqdocvar{c_2})) \wedge 
(\ell_2\ \notin (\coqdocvar{c_1} \Cup \coqdocvar{c_2}))
}
\end{theorem}
\coqdocnoindent

Note that to ensure that only a single pair of complementary literals
is deleted we need to assume that there is a unique complementary pair ($\mathit{UniqueCompPair}$). The above
theorem will not hold for the case with multiple complementary pairs. 

For the following theorem we need to assert in the assumption that for
any literal in one clause there exists no literal in the other clause
such that the sum of two literals is 0. This is defined by the
predicate \coqdocvar{NoCompLit}.
\begin{theorem}{All non-complementary, unequal literals are retained.}
\label{nocompl}
\coqdoceol
\coqdocnoindent
\ensuremath{
\hspace*{5mm}\forall \coqdocvar{c_1}\ \coqdocvar{c_2}\cdot\ 
\mathit{Wf}\ c_1\ \imp\ \mathit{Wf}\ c_2\ \imp\\
\hspace*{20mm}    \forall {\ell}_1\ \ell_{2}\cdot\
      (\ell_1 \in \coqdocvar{c_1})\ 
      \imp\ (\ell_2 \in \coqdocvar{c_2}) \imp\\
\hspace*{33mm}      (\coqdocvar{NoCompLit}\,\ell_1\,c_2) \imp
      (\coqdocvar{NoCompLit}\,\ell_2\,c_1)\ \imp \\ \hspace*{35mm}
      (\ell_1 \neq \ell_2)\ \imp\
(\ell_1 \in (\coqdocvar{c_1} \Cup \coqdocvar{c_2})) \wedge 
(\ell_2\ \in (\coqdocvar{c_1} \Cup \coqdocvar{c_2}))
}
\end{theorem}
\vspace*{-3mm}
\begin{theorem}{Only one copy of equal literals is retained (factoring).}\label{equal}
\coqdocemptyline
\coqdocnoindent
\ensuremath{
\hspace*{5mm}
\forall \coqdocvar{c_1}\ \coqdocvar{c_2}\cdot\ 
\mathit{Wf}\ c_1\ \imp\ \mathit{Wf}\ c_2\ \imp\\
\hspace*{20mm}    \forall {\ell}_1\ \ell_{2}\cdot\
      (\ell_1 \in \coqdocvar{c_1})\ \imp\ (\ell_2 \in \coqdocvar{c_2})\ 
      \imp\ (\ell_1 = \ell_2) \imp \\ 
\hspace*{35mm}((\ell_1 \in (\coqdocvar{c_1} \Cup \coqdocvar{c_2})) \wedge 
(\mathit{count}\,\, \ell_1\, (\coqdocvar{c_1} \Cup \coqdocvar{c_2}) = 1))
}
\end{theorem}
In order to check the resolution steps for each row, one has to
collect the actual clauses corresponding to their identifiers and this
is done by the \coqdocvar{findClause} function. 
The function  $\mathit{findClause}$ takes a list of clause identifiers
($\mathit{dlst}$), an accumulator ($\mathit{acc}$) to collect the list of
clauses, and requires as input a table that has the information about
all the input clauses ($\mathit{ctbl}$). If a clause id is processed, 
then its resolvent is fetched from the resolvent table
($\mathit{rtbl}$), else obtained from $\mathit{ctbl}$. If there is no
entry for a given id in the resolvent table and in the clause table,
an error is signalled. This error denotes the fact that there
was an input/output problem with the proof trace file due to which
some input clauses in the proof trace could not be accessed properly. 
This could have happened either because the proof trace was ill-formed 
accidentally or wilfully tampered with. 
\coqdocemptyline

The function that uses the $\Cup$ function recursively on a list of
input clause chain is called $\mathit{chainResolution}$ and it simply folds the
$\Cup$ function from left to right for every row in the proof part of
the proof trace file.
\coqdocemptyline
\hrule
\coqdocemptyline
\coqdocnoindent
\coqdockw{Definition} \coqdocvar{chainResolution} \coqdocvar{lst} =\coqdoceol
\coqdocindent{1.00em}
\coqdockw{match} (\coqdocvar{lst:list} (\coqdocvar{list} \coqdocvar{Z})) \coqdockw{with} \coqdoceol
\coqdocindent{2.00em}
\ensuremath{|} \coqdocvar{nil} \ensuremath{\Rightarrow} \coqdocvar{nil}\coqdoceol
\coqdocindent{2.00em}
\ensuremath{|} (\coqdocvar{x::xs}) \ensuremath{\Rightarrow}
\coqdocvar{List.fold\_left} ($\Cup$) \coqdocvar{xs} \coqdocvar{x}\coqdoceol
\coqdocemptyline
\hrule
\coqdocemptyline
\coqdocnoindent
The function $\mathit{findAndResolve}$ is our last function defined in
Coq world for Unsat checking and provides a wrapper on other
functions. Once the input clause file and proof trace files are opened
and read into different tables, $\mathit{findAndResolve}$ starts the checking
process by first obtaining the clause ids from the proof part of the
proof trace file, and then invoking $\mathit{findClause}$ to collect all the
clauses for each row in the proof part of the proof trace file. Once
all the clauses are obtained the function $\mathit{chainResolution}$ is called
and applied on the list of clauses row by row. For each row the
resolvent is stored in a separate table. The checker then simply
checks if the last row has an empty clause, and if there is one, it
agrees with the sat solver and says yes, the problem is UNSAT, else no.

If the proof trace part of the trace file contains nothing
(ill-formed) then there would be no entry for an identifier in the
trace table ($\mathit{ttbl}$), and this is signalled by en error state consisting
of a list with a single zero. Since zeros are otherwise prohibited to
be a legal part of the CNF problem description, we use them to signal
error states. Similarly, if the clause and trace table both are empty,
then a list with two zeros is output as an error.

\subsection{Program Extraction}
We extract the OCaml code by using the built-in 
extraction API in Coq. At the time of extraction
we mapped several Coq datatypes and data structures
to equivalent OCaml ones. 
For optimization we made the following replacements:
\begin{enumerate}
\item Coq Booleans by OCaml Booleans.
\item Coq integers (Z) by OCaml int.
\item Coq lists by OCaml lists.
\item Coq finite map by OCaml's finite map.
\item The combination of \coqdocvar {app} and \coqdocvar{rev} 
on lists in the function \coqdocvar{union}, and \coqdocvar{auxunion} 
was replaced by the tail-recursive  {List.rev\_append} in OCaml.
\end{enumerate}
Replacing Coq Zs with OCaml integers gave a performance boost by a
factor of 7-10. Making minor adjustments by replacing the Coq finite maps
by OCaml ones and using tail recursive functions gave a further 20\% 
improvement. An important consequence of our extraction is that only
some datatypes, and data structures get mapped to OCaml's; the key
logical functionality is unmodified. The decisions for making changes
in data types and data structures are a standard procedure in any
extraction process using Coq~\cite{CoqBook}.

\begin{center}
\begin{table*}[!ht]
\begin{flushleft}
\captionsetup{justification=raggedright}\caption{Comparison of our results with HOL 4 and
  Tracecheck. Number of Resolutions (inferences) shown for HOL 4 is the number
  that HOL 4 calculated from the proof trace obtained from running
  ZVerify - the uncertified checker (for zChaff) that Amjad used 
to obtain the proof trace. \cunsat Resolution count is obtained from
the proof trace generated by the uncertified checker Tracecheck. 
In terms of inferences/second, we are 1.5 to 32 times faster
than Amjad's HOL 4 checker, whilst a factor 2.5 slower than
Tracecheck. All times shown are total times for all the three
checkers.The symbol z? denotes that zChaff timed out after an hour.
\vspace{1.0ex}
}\label{comptbl}
\end{flushleft}
\resizebox{!}{5.5cm}{
\hfill{}
\begin{tabular}{|r|l|r|r|r|r|r|r|r|r|l|r|} \hline \hline
{No.}  & {Benchmark} & 
\multicolumn{3}{c|}{{HOL 4}} & 
\multicolumn{3}{c|} {{\cunsat}} & 
\multicolumn{2}{c|}{Tracecheck} \\ \hline
& & \emph{Resolutions} & \emph{Time} & \emph{inf/sec} & \emph{Resolutions}
 & \emph{Time} & \emph{inf/s} & \emph{Time} & \emph{inf/s} \\ \hline
1. & een-tip-uns-numsv-t5.B & 89136 & 4.61 & 19335 & 122816 &  0.86 & 142809 & 0.36 & 
    341155 \\ \hline
2. & een-pico-prop01-75 & 205807 & 5.70 &  36106 & 246430 &  1.67 &
147562 & 0.48 & 513395 \\ \hline
3. & een-pico-prop05-50 & 1804983  & 58.41 &  30901 & 2804173 &  20.76
& 135075 & 8.11 & 345767\\   \hline  
4. & hoons-vbmc-lucky7 & 3460518 & 59.65 & 58013 &  4359478 &  35.18 &
123919 & 12.95 & 336639 \\ \hline
5. & ibm-2002-26r-k45 & 1448 &   24.76 &  58 & 1105 &  0.004 &  276250
& 0.04 & 27625 \\ \hline
6. & ibm-2004-26-k25 & 1020 &  11.78 & 86 &  1132 &  0.004 & 283000 &
0.04 & 28300 \\ \hline
7. & ibm-2004-3\_02\_1-k95 & 69454 &  5.03 & 13807 & 114794 &  0.71 &
161681 & 0.35 & 327982 \\ \hline
8. & ibm-2004-6\_02\_3-k100 &  111415 & 7.04 & 15825 & 126873 & 0.9 &
140970 & 0.40 & 317182 \\ \hline
9. & ibm-2002-07r-k100 & 141501 & 2.82 &  50177 & 255159  & 1.62   &
157505 & 0.54 & 472516 \\ \hline
10. & ibm-2004-1\_11-k25 & 534002 &  13.88 & 38472 & 255544  &  1.77 & 144375 
& 0.75 & 340725 \\ \hline
11. & ibm-2004-2\_14-k45 & 988995 &  31.16 & 31739 & 701430  &  5.42
& 129415 & 1.85 & 379151 \\ \hline 
12. & ibm-2004-2\_02\_1-k100 & 1589429 &  24.17 & 65760 & 1009393 &
7.42 & 136036 & 3.02 &  334236 \\ \hline
13. & ibm-2004-3\_11-k60  &  z? & z? & - & 13982558 & 133.05 & 105092
& 59.27 & 235912 \\ \hline
14. & manol-pipe-g6bi &	82890 & 2.12 & 39099 & 245222 & 1.59 & 154227 &
 0.50 &  490444 \\ \hline
15. & manol-pipe-c9nidw\_s & 700084  & 26.79 &  26132 & 265931 &  1.81
& 146923 & 0.54  & 492464 \\ \hline
16. & manol-pipe-c10id\_s & 36682 & 11.23 & 3266 & 395897 & 2.60 &
152268 & 0.82 &  482801\\ \hline
17. & manol-pipe-c10nidw\_s & z? &z? & - & 458042 & 3.06 & 149686 & 1.21  
& 381701 \\ \hline
18. & manol-pipe-g7nidw & 325509  &  8.82 & 36905 & 788790 & 5.40 
 & 146072 & 1.98 & 398378\\ \hline
19. & manol-pipe-c9 & 198446 & 3.15 &  62998 & 863749 &   6.29 &
137320 & 2.50 & 345499 \\ \hline
20. & manol-pipe-f6bi & 104401 &  5.07 &  20591 & 1058871 &  7.89 &
134204 & 2.97  & 356522 \\ \hline
21. & manol-pipe-c7b\_i & 806583 &   13.76 & 58617 & 4666001 & 38.03 &
122692 & 15.54 & 300257\\ \hline
22. & manol-pipe-c7b & 824716 & 14.31 & 57632 & 4901713 &  42.31 &
115852 & 18 & 272317\\ \hline
23. & manol-pipe-g10id & 775605 &  23.21 &  33416 & 6092862 &  50.82 &
119891 & 21.08 & 289035\\ \hline
24. & manol-pipe-g10b	& 2719959 &  52.90 &  51416 & 7827637 &  64.69 & 121002
& 26.85 & 291532 \\ \hline
25. & manol-pipe-f7idw & 956072 &  35.17 & 27184 & 7665865 & 68.14 &
112501 & 30.74 & 249377 \\ \hline
26. & manol-pipe-g10bidw &  4107275  & 125.82  & 32644  & 14776611   
& 134.92 & 109521 & 68.13  & 216888\\ \hline
\end{tabular}}
\end{table*}
\end{center}

\section{Experimental Results}
We evaluated our certified checker \cunsat on a set of benchmarks from 
the SAT Races of 2006 and 2008 and the SAT Competition of 2007. We
present our results on a sample of the SAT Race Benchmarks in
Table~\ref{comptbl}. The results for \cunsat shown in the table are
for validating proof traces obtained from the PicoSAT solver. 
Our experiments were carried out on a server
running Red Hat on a dual-core 3 GHz, Intel Xeon CPU with 28GB
memory. 

The HOL 4 and Isabelle based checkers~\cite{Weber:2009p1923} were
evaluated on the SAT Race Benchmarks shown in the
table~\cite{weber-private09}.
Isabelle reported segmentation faults on most of the problems, whilst
HOL 4's results are summarized along with our's in
Table~\ref{comptbl}. HOL 4 was run on an AMD dual-core 3.2 GHz
processor running Ubuntu with 4GB of memory.
We also compare our timings with that obtained from the uncertified checker Tracecheck. Since the size of the proof
traces obtained from zChaff is substantially different than the size
of the traces obtained from Tracecheck on most problems, we decided to
compare the speed of our checker with HOL 4 and Tracecheck in terms of
resolutions (inferences) solved per second. We observe that in terms
of inferences/sec we are 1.5 to 32 times faster than HOL 4 and
2.5 times slower than Tracecheck. Times shown for all the three checkers
in the table are the total times  including time spent on actual
resolution checking, file I/O and garbage collection. Amjad reported
that the version of the checker he has used on these benchmarks is
much faster than the one published in~\cite{Weber:2009p1923}. 

As a proof of concept we also validated the proof traces from zChaff by
translating them to PicoSAT's trace format. The performance of \cunsat
in terms of {\it inf/sec} on the translated proof traces (from zChaff
to PicoSAT) was similar to the performance of \cunsat
when it checked PicoSAT's traces obtained directly from the PicoSAT
solver -- something that is to be expected.
\vspace*{-3mm}
\subsection{Discussion}
The Coq formalization consisted of 8 main function definitions
amounting to nearly 160 lines of code, and 4 main theorems shown in
the paper and 4 more that are about maps (not shown here due to
space). Overall the proof in Coq was nearly
4000 lines consisting of the proofs of several big and small lemmas
that were essential to prove the 4 main theorems. The extracted
OCaml code was approximately 2446 lines, and the OCaml glue code was 324 lines. 

We found that there is no
implementation of the array data type which meant that we had to use
the type of list. Since lists are defined inductively, it is easier to
do reasoning with them,  although implementing very fast and
efficient functions on these is impossible. 
In a recent development related to Coq, there has been an emergence of
a tool called Ynot~\cite{Nanevski:2008p1326} that can deal with
arrays, pointers and file
related I/O in a Hoare Type Theory. Future work in certification using
Coq should definitely investigate the relevance and use of this.

We noticed that the OCaml
compiler's native code compilation does produce efficient binaries but
the default settings for automatic garbage collection were not
useful. We observed that if we do not tune the runtime environment
settings of OCaml by setting the values of {\tt{\small{OCMALRUNPARAM}}}, as soon as
the input proof traces had more than a million inferences, garbage
collection would kick in so severely that it will end up consuming (and
thereby delaying the overall computation) as much as 60\% of the total
time. By setting the initial size of major heap to a large value such
as 2 GB and making the garbage collection less eager, we noticed that
the computation times of our checker got reduced by upto a factor of 7 on
proof traces with over 1 million inferences.

\vspace*{-2mm}\section{Related Work}

Recent work on checking the result of SAT solvers can be traced to the
work of Zhang \& Malik~\cite{zhang-date03} and Goldberg \&
Novikov~\cite{goldberg-date03}, with additional insights provided in
recent work~\cite{Beame:2004p1080,vangelder-sat07}.
Besides Weber and Amjad, others who have advocated the use of a checker
include Bulwahn et al~\cite{Bulwahn:2008p6003} who experimented with the
idea of doing reflective theorem proving in Isabelle and suggested that
it can be used for designing a SAT checker. 

In a recent paper~\cite{Maric:2009p6673}, Mari{\'c} presented a
formalization of SAT solving algorithms in Isabelle that are used in
modern day SAT solvers. 
An important difference is that 
whereas we have formalized a SAT checker and {\em extracted} an
executable code from the formalization itself, Mari{\'c} formalizes a
SAT solver (at the abstract level of state machines) and then
implements  the  verified algorithm in the SAT solver {\em off-line}.

An alternative line of work involves the formal development of SAT
solvers. Examples include the work of Smith \&
Westfold~\cite{smith-tr08} and the work of Lescuyer and
Conchon~\cite{Lescuyer:2009p1671}. Lescuyer and Conchon have
formalized a simplified SAT solver in Coq and extracted an executable.
However, the performance results have not been reported on any
industrial benchmarks. This is because they have not formalized
several of the key techniques used in modern SAT solvers.
The work of of Smith \& Westfold involves the formal synthesis of a SAT
solver from a high level description. Albeit ambitious, the
preliminary version of the SAT solver does not include the most
effective techniques used in modern SAT solvers.

There has been a recent surge in the area of certifying SMT solvers.
M.~Moskal recently provided an efficient certification technique for
SMT solvers~\cite{Moskal:2008p1806} using term-rewriting systems. The
soundness of the proof checker is guaranteed through a formalization
using inference rules provided in a term-rewriting formalism. L.~de
Moura and N.~Bj\o rner~\cite{deMoura:2009p1840} presented the proof
and model generating features of the state-of-the-art SMT solver Z3.

\vspace*{-4mm}\section{Conclusion}
In this paper we presented a methodology for performing efficient yet
formally certified SAT solving. The key feature of our approach is
that we did a one-off formal design and reasoning of the checker using
Coq proof-assistant and extracted an OCaml program which was used as a
standalone executable to check the outcome of industrial-strength
SAT solvers such as PicoSAT and zChaff. Our certified checker can be plugged in
with any proof generating SAT solver with previously agreed
certificates for satisfiable and unsatisfiable problems. On one hand
our checker provides much higher assurance as compared to uncertified
checkers such as Tracecheck and on the other it enhances usability
and performance when compared to the certified checkers implemented in
HOL 4 and Isabelle. In this regard our approach provides an arguably
optimal middle ground between the two extremes. We are investigating
on optimizing the performance aspects of our checker even further so
that the slight difference in overall performance between uncertified
checkers and us can be further minimized.

\vspace*{-2mm}
\medskip\noindent
{\small
{\bf Acknowledgements.~}
We thank H.~Herbelin, Y.~Bertot, P.~Letouzey, and many more people on the Coq mailing
list who helped us with Coq questions. We also thank T.~Weber and
H.~Amjad for answering our questions on their work and also carrying
out industrial benchmark evaluation on their checker.  A.~P.~Landells
helped out with server issues.  This work was partially funded by
EPSRC Grant EP/E012973/1, and EU Grants ICT/217069 and IST/033709.
}



\end{document}